\documentclass[11pt]{amsart}
\usepackage{latexsym,amssymb,amsmath,youngtab}
\usepackage{youngtab}
\usepackage{epsfig}
\usepackage{amsmath,amsthm,amssymb,amscd,MnSymbol}
\usepackage[mathscr]{eucal}
\usepackage{verbatim}
\usepackage{tikz}
\usepackage{caption}
\usepackage{subcaption}
\usepackage[arrow,matrix,graph,frame,poly,arc,tips]{xy}

\usepackage{color}
\newtheoremstyle{custom}
  {3pt}
  {3pt}
  {\slshape}
  {}
  {\bfseries}
  {.}
  { }
   {}
\theoremstyle{custom}
\newtheorem{theorem}{Theorem}[section]
\newtheorem{proposition}[theorem]{Proposition}

\newtheorem{proposition/definition}[theorem]{Proposition/Definition}
\newtheorem{lemma}[theorem]{Lemma}
\newtheorem{corollary}[theorem]{Corollary}

\theoremstyle{definition}

\theoremstyle{remark}
\newtheorem{remark}[theorem]{Remark}

\newtheoremstyle{exercise}
  {3pt}
  {6pt}
  {}
  {}
  {\bfseries}
  {:}
  { }
   {}
\theoremstyle{exercise}
\newtheorem{exercise}[theorem]{Exercise}
\newtheoremstyle{exercises}
  {3pt}
  {6pt}
  {}
  {}
  {\bfseries}
  {:}
  {\newline}
   {}

\theoremstyle{exercise}
\newtheorem{exercises}[theorem]{Exercises}

\input epsf
\def\boxit#1{\vbox{\hrule height1pt\hbox{\vrule width1pt\kern3pt
  \vbox{\kern3pt#1\kern3pt}\kern3pt\vrule width1pt}\hrule height1pt}}

\def\BC{\mathbb C}\def\BF{\mathbb F}
\def\BH{\mathbb H}

\def\tdim{{\rm dim}}

\def\hd{,...,}
\def\ww{\wedge}

\def\11{\mathbf 1}

\def\a{\alpha}

\def\b{\beta}

\def\s{\sigma}
\def\k{\kappa}

\def\ot{{\mathord{ \otimes } }}

\def\ra{{\mathord{\;\rightarrow\;}}}

\def\dim{{\rm dim}\;}

\def\BF{\Bbb F}\def\BH{\Bbb H}

\def\s{\sigma}

\def\a{\alpha}
\def\b{\beta}

\def\FS{\mathfrak  S}

\def\BC{\mathbb  C}

\def\ci{\mathcal  I}

\def\hd{, \hdots ,}

\def\ra{\rightarrow}

\def\tperm{\operatorname{perm}}

\def\tdim{\operatorname{dim}}

\def\ww{\wedge}

\def\be{\begin{equation}}
\def\ene{\end{equation}}

\def\HP{HP}\def\HF{HF}

\begin{document}
\title[Minimal free resolutions of sub-permanents]  
{On minimal free resolutions of sub-permanents \\ and other ideals
arising in  complexity theory}

\author{Klim Efremenko}
\address{Department of Computer Science, Ben-Gurion University}
\email{klimefrem@gmail.com}

\author{J.M. Landsberg}
\thanks{Landsberg supported by NSF  DMS-1405348}
\address{Department of Mathematics, Texas A\&M University}
\email{jml@math.tamu.edu}

\author{Hal Schenck}
\thanks{Schenck supported by NSF DMS-1312071}
\address{Department of Mathematics, Iowa State University, Ames, Iowa 50011}
\email{hschenck@iastate.edu}

\author{Jerzy Weyman}
\thanks{Weyman supported by NSF DMS-1400740}
\address{Department of Mathematics, University of Connecticut}
\email{jerzy.weyman@gmail.com}


 \begin{abstract}
 We compute the linear strand of the minimal free resolution of the ideal generated by $k\times k$ sub-permanents of an
 $n\times n$ generic matrix and  of the ideal generated by square-free monomials of degree $k$.
 The latter calculation gives the full minimal free resolution by \cite{MR2371263}. 
 Our motivation is  to lay   groundwork for  the use of commutative algebra
 in algebraic complexity theory.
 We also compute several  Hilbert functions 
 relevant for complexity theory.   
\end{abstract}

\maketitle

\section{Introduction}

We study homological properties of two families of ideals over polynomial rings:
the  ideals $\ci^{sqf;n,k}\subset  \BC [x_1,\ldots ,x_n]$ generated by square-free monomials of degree $k$ in $n$ variables and
the    ideals $\ci^{perm_n,k}\subset  \BC[x_{i,j}]_{1\le i,j\le n}$ generated by $k\times k$ sub-permanents of an $n\times n$ generic matrix.
Recall that the  permanent of an $m\times m$ matrix $Y=(y_{i,j})$ is the polynomial
$$perm_m(Y) =\sum_{\s\in\FS_m} y_{1,\s(1)}y_{2,\s(2)}\cdots y_{m,\s(m)},$$
where $\FS_m$ denotes the symmetric group on $m$ elements.

  We obtain our results via   larger  ideals $I_{1\times k}(1,n)$  (resp. $I_{1\times k}(n,n)$).
The ideal $I_{1\times k}(1,n)$ is generated by all monomials of degree $k$ in $n$ variables. 
The ideal $I_{1\times k}(n,n)$ is generated by permanents of $k\times k$ matrices
produced from $X$ where   repetition of  rows and columns is allowed.
Invariantly, $I_{1\times k}(1,n)=\oplus_{j\geq k} S^j\BC^n$ is the ideal generated by $S^k\BC^n$
and $I_{1\times k}(n,n)\subset Sym(\BC^{n^2})$ is the ideal generated $S^k\BC^n\ot S^k\BC^n\subset S^k(\BC^{n }\ot \BC^n)$.
The main result in each case says that the linear strand of resolution of $\ci^{sqf;n,k}$ (resp. $\ci^{perm_n,k}$) 
is the  subcomplex of the  linear strand of the resolution of $I_{1\times k}(1,n)$
(resp. of $I_{1\times k}(n,n)$) consisting of elements of {\it regular weights} (cf. \S 2). 

Our motivation comes from  complexity theory. 
We seek to find   differences between the homological behavior of ideals generated by $k\times k$ minors
(i.e., subdeterminants) of the generic matrix and the ideals generated by $k\times k$ subpermanents. 
The ideal generated by square-free monomials arises as the $(n-k)$-th Jacobian ideal of the monomial $x_1x_2\ldots x_n$.

\subsection*{Acknowledgments} We thank the anonymous referee for numerous useful suggestions, including 
simplifications of several proofs. Efremenko, Landsberg and Weyman thank
the Simons Institute for the Theory of Computing, UC Berkeley, for providing a wonderful 
environment during the fall 2014 program {\it Algorithms and Complexity in Algebraic Geometry} to work on  this article.

\section{Preliminaries}

\subsection{Representation Theory   }\label{subsec:repthy} For proofs of the statements here, see, e.g., \cite{FH} or \cite[Ch.~2]{weyman}.
We work exclusively over the complex numbers $\BC$, although our results hold for an arbitrary field 
 of characteristic $0$. If $W$ is a $\BC$--vector space of dimension $n$, a choice of basis determines
 a maximal torus of diagonal matrices and a labeling of weights for the torus
 by  $n$-tuples $\lambda=(\lambda_1,\cdots,\lambda_n)\in\mathbb{Z}^n$. A weight  $\lambda$ is   {\it dominant } if $\lambda_1\geq\lambda_2\geq\cdots\geq\lambda_n$. Irreducible representations of $GL(W)$ are in one-to-one correspondence with dominant weights $\lambda$. Let $S_{\lambda}W$ denote the irreducible representation associated to $\lambda$.
 Write $|\lambda|=\lambda_1+\cdots+\lambda_n$
 for the size of $\lambda$.  A weight $\alpha =(\alpha_1,\ldots ,\alpha_n)$ is {\it regular} if each $\alpha_i$ is equal to $0$ or $1$.
The regular weights will play an important role in stating our results.

 When $\lambda$ is a dominant weight with $\lambda_n\geq 0$, we say that $\lambda$ is a {\it partition} of $r=|\lambda|$, and we write $\lambda\vdash r$.
 When   dealing with partitions we often omit the trailing zeros. Associated to a partition is its  {\it Young diagram} which consists of left--justified rows of boxes, with $\lambda_i$ boxes in the  $i$--th row: for example, the Young diagram associated to $\lambda=(5,2,1)\vdash 8$ is
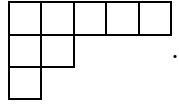
\[\Yvcentermath1\yng(5,2,1).\]
 The {\it transpose} $\lambda'$ of a partition $\lambda$ is obtained by transposing the corresponding Young diagram. For the example above, $\lambda'=(3,2,1,1,1)$.

 Given finite dimensional $\BC$--vector spaces $F,G$,
 the {\it Cauchy formulas}   describe the decomposition of the symmetric and exterior powers of $F\otimes G$ into a sum of irreducible $GL(F)\times GL(G)$--representations, see, e.g.,  \cite[Cor.~2.3.3]{weyman}:
\begin{equation}\label{eq:cauchy}
\begin{aligned}
 Sym^n(F\otimes G)=\bigoplus_{\lambda\vdash n} S_{\lambda}F\otimes S_{\lambda}G, \\
  \bigwedge^n(F\otimes G)=\bigoplus_{\lambda\vdash n} S_{\lambda}F\otimes S_{\lambda'}G.
\end{aligned}
\end{equation}

Let  $I_{a\times b}(m,n)$ denote the ideal   generated by $S_{b^a} \BC^m\otimes S_{b^a}\BC^n$.

\subsection {$GL_m$ and $\FS_m$ representations}
Let  $ E=\BC^m$ equipped with its standard basis. The symmetric group $\FS_m$  is then contained in $GL(E)$ as the permutation matrices.
Consider the irreducible representation $S_\lambda E$ where $\lambda$ is a partition of $m$. Inside of $S_\lambda E$ we have the   $\FS_m$-submodule spanned by the  elements
of   regular weight $(1^m)=(1,1,\ldots ,1)$. This submodule is denoted $[\lambda]$, the Specht module corresponding to $\lambda$. The
representations $[\lambda]$  are the distinct  irreducible representations of $\FS_m$ (see, e.g.,  \cite{MR0414794}).
Write 
$$(S_\lambda \BC^m)_{reg}=[\lambda].$$

Recall that for finite groups $H\subset G$, and an $H$-module $W$,
$Ind^G_H W=\BC[G]\ot_{\BC[H]}W$ is the induced $G$-module.
For $n\ge m$ and   $\lambda\vdash m$, 
$$(S_\lambda \BC^n)_{reg}\equiv Ind_{\FS_m\times \FS_{n-m}}^{\FS_n} ([\lambda]\otimes [n-m]).$$

Introduce the notation
$\tilde \FS_k=\FS_k\times \FS_{n-k}\subset \FS_n$, and if $\pi$ is a partition of $k$, 
write $\widetilde{[\pi]}=[\pi]\times [n-k]$ for the $\tilde\FS_k$ module that is $[\pi]$ as an $\FS_k$-module and trivial
as an $\FS_{n-k}$-module.

\subsection{Finite free resolutions}
 Let $S=\BC[x_1,\ldots ,x_N]$  be the ring of polynomials in $N$ variables
equipped with its grading by degree.
Let  $S_i$ denote its $i$-th graded component.
Let $S_+$ denote the maximal ideal $S_+=\oplus_{i>0}S_i$.
Let $M=\oplus_{i\ge 0}M_i$ be a graded $S$-module.
A complex of free graded $S$-modules
$${\mathbb F}:0\ra \mathbb{F}_N\buildrel{d_N}\over\ra \mathbb{F}_{N-1}\ra\cdots\ra \mathbb{F}_1\buildrel{d_1}\over\ra \mathbb{F}_0$$
is a minimal free resolution of $M$ if the only homology of $\mathbb F$ is $H_0(\mathbb{F})=M$,
and $d_i$ are maps of degree $0$ such that $d_i(\mathbb{F}_i)\subset S_+\mathbb{F}_{i-1}$ for all $i$.

Define the {\it graded Betti numbers} $\beta_{i,j}$ of $M$ by  
$$\mathbb{F}_i=\oplus_{j\ge 0} S(-i-j)^{\beta_{i,j}}$$
where $S(-m)$ denotes a copy of $S$ with generator in degree $m$.

If $M$ is an ideal   generated in degree $k$, then  $\mathbb{F}_i=\oplus_{j\ge k} S(-i-j)^{\beta_{i,j}}$.
  The linear strand of $\mathbb F$ is a subcomplex

$${\mathbb F}^{lin}:0\ra \mathbb{F}^{lin}_N\buildrel{d_N}\over\ra \mathbb{F}^{lin}_{N-1}\ra\cdots\ra \mathbb{F}^{lin}_1\buildrel{d_1}\over\ra \mathbb{F}^{lin}_0$$
where $\mathbb{F}^{lin}_i=S(-i-k)^{\beta_{i,i+k}}$.
The graded Betti numbers $\beta_{i,j}$ have the interpretation in terms of $Tor$ functors
$$\beta_{i,j}=dim_{\BC} Tor^S_i (S/S_+ ,M)_{i+j}$$
where the subscript denotes the homogeneous component.
This means that $\beta_{i,j}$ can be calculated as 
\begin{equation}
Tor^S_i (S/S_+, M)_{i+j}=H_i (\mathbb C(x_1,\ldots ,x_N; M))_{i+j}\label{torformula}
\end{equation}
where 
$\mathbb C(x_1,\ldots ,x_N;M)$ is the {\it  Koszul complex} 
defined by
$\mathbb C(x_1,\ldots ,x_N;M)_i=\bigwedge^i\BC^N\ot_S M $ and the differential
$
d: \bigwedge^i\BC^N\ot_S M \ra \bigwedge^{i-1}\BC^N\ot_S M $
by 
$$
d(e_{j_1}\ww\cdots \ww e_{j_i}\ot m)=\sum_{u=1}^i(-1)^{u+1} e_{j_1}\ww\cdots \ww \hat e_{j_u}\ww\cdots \ww e_{j_i}\ot x_{j_u}m.
$$

\section{The linear strands of the minimal free resolutions of the ideals   $\ci^{perm_n,k}$}

\subsection {The resolution of $I_{1\times k}(n,n)$}

The ideal $I_{1\times k}(n,n)$ is  the ideal generated by $S_kE\otimes S_kF\subset S^k(E\ot F)$, where $E,F\simeq \BC^n$.
The minimal free resolution of this ideal is   known (see, e.g.,  \cite{MR3565359}, where it is denoted by $I_{1\times k}$).
The linear components  of this resolution are generated by 

\begin{equation}
\label{sreslin}
\underline {\mathbb {F}}_j^{lin}= \bigoplus_{a+b=j}S_{(k+b,1^{a}) }E.
\end{equation}
So the $j$-th linear term is $\mathbb {F}_j^{lin}=\underline {\mathbb {F}}_j^{lin}\otimes S_{( k+a,1^{b}) }F\otimes S(-k-j)$.

Since the resolution is $GL(E)\times GL(F)$-equivariant, each module in the complex has a double weight decomposition  induced by the restricted action of pairs of diagonal matrices.

\subsection {The main result}

We work over $S=\BC[x_{i,j}]_{1\le i,j\le n}=Sym(E\otimes F)$.

Define a sub-complex $\mathbb {H}^{lin}$ of the complex $\mathbb{F}^{lin}$ given by   \eqref{sreslin}
by setting   $\underline{\BH}_j^{lin}$ to be the subspace  of $\underline{\BF}_j^{lin}$ spanned  by the basis elements of regular content.
Note  that $\mathbb {H}^{lin}$ is indeed a sub-complex of $\mathbb{F}^{lin}$.
Let $E_j\subset E$, $F_j\subset F$ denote the span of the first $j$ basis vectors.

\begin{theorem} \label{mjthm}
When $k>1$, the complex $\mathbb{H}^{lin}$ is the  linear part of the minimal free resolution of the ideal $\ci^{perm_n,k}$.
Moreover, $\tdim \underline{\mathbb{H}}_{j}={\binom{n}{\k+j }}^2\binom {2(\k+j-1)}{j }$.

As an $\FS_n\times \FS_n$-module,  
\be 
\underline{\mathbb{H}}^{lin}_{j}= Ind_{\tilde\FS_{ {\k+j}}\times\tilde\FS_{ {\k+j }}}^{\FS_{ n}\times \FS_{ n}}
(\bigoplus_{a+b=j}\widetilde{[\k+b ,1^{a}]}_{E_{\k+j}}\ot \widetilde{[\k+a,1^b]}_{F_{\k+j }}) .
\ene
\end{theorem}

\begin{proof}
Consider  the complex $\mathbb{F}^{lin}$ giving the  linear strand of the ideal $I_{1\times k}(n,n)$.
The term $\mathbb{F}_0$ consists of  the generators of $I_{1\times k}(n,n)$,  the space $S_kE\otimes S_kF$.

 Inside $S_k E\ot S_k F$ is the ideal generated by the sub-permanents 
 which consists of the subspace of 
 regular weights.
  Note that the set of regular vectors in any
 $E^{\ot m}\ot F^{\ot m}$ (where we assume $m\leq n$ in order for the set of such vectors to be nonempty) spans a   $\FS_E\times \FS_F$-submodule.
 
  The linear strand of the $j$-the term in the minimal free resolution
 of the ideal $\ci^{perm_n,k}$  is also   a 
 $\FS_E\times \FS_F$-submodule of
 $\mathbb{F}_j$.
 We claim this sub-module is generated by  the span of the regular vectors.  
 In what follows   $p(i_1,\ldots ,i_k;j_1,\ldots , j_k)$ denotes the sub-permanent formed from   rows $i_1,\ldots i_k$ and columns $j_1,\ldots ,j_k$.

We work by induction, the case $j=0$ was discussed above.
Assume the result has been proven up to homological degree $j-1$ and consider  the homological degree $j$ and homogeneous degree $k+j$.
The generators of the $j$-th module in the linear strand of the resolution of $\ci^{perm_n,k}$ have to be contained in linear part of 
 $\mathbb{H}^{lin}_{j-1}$, so all its
weights are either regular, or  such that one of the row indices $i_{\a}$  is $2$, and/or one
of the column indices $j_{\b}$  is $2$, and all other $p_u,q_u$ are zero or $1$. Call such 
a weight {\it sub-regular}. It remains to show that no linear syzygy with a sub-regular
weight can appear.
To do this we show that no sub-regular weight vector in 
$(\mathbb{F}_{j})_{subreg}$    maps to zero
in $(\mathbb{H}_{j-1})\cdot (E\ot F)$.

First consider the case where both the $E$ and $F$
weights are sub-regular, then (because the space is a $\FS_E\times \FS_F$-module),
the weight
$(2,1\hd 1,0\hd 0)\times (2,1\hd 1,0\hd 0)$ must appear in the syzygy.
The only  way for this to appear is to have a term  divisible by $x_{1,1}$.
But,  since $x_{1,1}$
is not a  zero-divisor in $Sym(V)$, such a term  cannot map to zero  because our syzygy is a syzygy of degree zero
multiplied by $x_{1,1}$. But by minimality no such syzygy  exists.

Finally consider the case where there is
a vector of weight $(2,1^{j+k-2})\times (1^{j+k})$ appearing.
Here it is more convenient to look at the calculation of the free resolution using the Koszul complex.
Such a syzygy would give   a Koszul cycle with   summands of the form
\begin{equation}\label{sum}
z=\sum_{t}a_t  e_{a_{1,t},b_{1,t}}\wedge\cdots\wedge e_{a_{j,t},b_{j,t}}\otimes p(I_t;J_t)
\end{equation}
where $p(I_t;J_t)$ are subpermanents formed from distinct rows and columns and $a_t\in \BC$. 
The total weight is $(2,1^{k+j-2})\times (1^{k+j})$ and  the Koszul differential $d(z)$ is zero.
Consider  this differential. The coefficients of all the basis elements $e_{a_1,b_1}\wedge\cdots\wedge e_{a_{j-1},b_{j-1}}$ of $d(z)$ have to be zero.
There are  three kinds of basis elements: the indices $a_1,\ldots ,a_{j-1}$ can contain number $1$ twice, once or not contain $1$ at all.
Consider  the basis elements not containing $1$, say the element $e_{2,1}\wedge\cdots\wedge e_{j,j-1}$. The only elements that can appear on  
the right hand side of the tensor product in $d(z)$
are the elements $x_{1,s}p(1,j+1,j+2,j+k-1;j,j+1,j+2,\ldots ,\hat{s},\cdots ,j+k)$, for $s=j, j+1,\ldots, j+k$.

\begin{lemma}\label{x1slem}  Let $k>1$. The elements $x_{1,s}p(1,j+1,j+2,j+k-1;j,j+1,j+2,\ldots ,\hat{s},\cdots ,j+k)$, for $s=j, j+1,\ldots, j+k$ are linearly independent in $S$.
\end{lemma}

\begin{proof} 

After re-labeling, the lemma  amounts to showing the polynomials $x_{1,1}P_1,\ldots , x_{1,k+1}P_{k+1}$ are linearly independent, where
$P_i$ is the permanent of the matrix obtained by removing the $i$-th column of
$$\left(\begin{matrix} x_{1,1}&x_{1,2}&\ldots&x_{1,k+1}\\
x_{2,1}&x_{2,2}&\ldots&x_{2,k+1}\\
\ldots&\ldots&\ldots&\ldots\\
x_{k,1}&x_{k,2}&\ldots&x_{k,k+1}
\end{matrix}\right).$$
Say that
\begin{equation}\label{sum2}
\sum_{s=1}^{k+1}  b_s x_{1,s} P_s=0
\end{equation}
for some scalars $b_s$. We need to show that all $b_s$ are zero. By symmetry it suffices to show that $b_1=b_2=b_3=0$. So   set $x_{1,s}=0$ for $s>3$.
Using the Laplace expansion of permanents along the first row, and writing $P_{i,j}$ for  the permanent obtained by 
removing row $1$ and columns $i,j$ we can rewrite (\ref{sum2}) as
$$b_1x_{1,1}(x_{1,2}P_{1,2}+x_{1,3}P_{1,3})+b_2x_{1,2}(x_{1,1}P_{1,2}+x_{1,3}P_{1,3})+b_3(x_{1,1}P_{1,3}+x_{1,2}P_{2,3})=0$$
which gives $b_i+b_j=0$ for $1\le i<j\le 3$,  so $b_1=b_2=b_3=0$.
\end{proof}

Lemma \ref{x1slem} implies  that in all the summands in $z$  of \eqref{sum} with $a_t\ne 0$,  both 
appearances of the index $1$ have to occur among $a_{1,t},\cdots ,a_{j,t}$. 
Now consider     the coefficient of the basis element where $1$ occurs  among $a_1,\ldots ,a_{j-1}$, say, $e_{1,1}\wedge e_{2,2}\wedge\cdots \wedge e_{j-1,j-1}$ in $d(z)$.
We obtain a linear combination of the elements $x_{1,s} p(j,j+1,\cdots ,j+k-1; j,j+1,\cdots ,\hat{s},\cdots ,j+k)$ for $s=j,j+1,\ldots ,j+k$. But  these elements are trivially linearly independent (all monomials occurring in them are different) so all coefficients $a_t$ are zero.

The rest of Theorem \ref{mjthm} follows because
  if $\pi$ is a partition of $ {\k+j}$ then
the weight $(1\hd 1)$ subspace of $S_{\pi}E_{\k+j}$, considered
as an $\FS_{  E_{\k+j}}$-module, is $[\pi]$ (see, e.g., \cite{MR0414794}), and the space of regular vectors
in $S_{\pi}E\ot S_{\mu}F$  is $Ind_{\tilde\FS_{E_{\k+j}}\times \tilde\FS_{F_{\k+j}}}^{\FS_E\times \FS_F}
\widetilde{[\pi]}_E\ot \widetilde{[\mu]}_F$. In the formula for the
dimension of $\underline{\mathbb{H}}_{j}$, the factor  ${n\choose {k+j}}^2$ is explained by inducing. The 
 dimensions of hook Specht modules are binomial coefficients, $\dim [x,1^y]={{x+y-1}\choose y}$
So we need to prove that 
$$\sum_{a+b=j}{{k+j-1}\choose a}{{k+j-1}\choose b}= {{2(k+j-1)}\choose j}$$
This has a combinatorial explanation. Given $2(k+j-1)$ balls, 
$k+j-1$ white and $k+j-1$ black, both sides of the equation calculate number of 
choices of $k+j-1$ of them: the left side partitions into how 
many white (a) and how many black (b) are chosen. 
\end{proof}

\begin{remark}
For small $n$ and $\k$, computer computations show no additional first syzygies on the $\k \times \k$ sub-permanents
of a generic $n \times n$ matrix (besides the linear syzygies) in degree less than the
 degree of the Koszul relations $2\k$. For example, for  $\k=3$ and $n=5$, there are $100$ cubic generators for the ideal
and $5200$ minimal first syzygies of degree six. There can be at most $\binom{100}{2} = 4950$ 
Koszul syzygies, so there must  be additional non-Koszul first syzygies.
\end{remark} 

\section{The  minimal free resolutions of the ideals   $\ci^{sqf;n,k}$}

\subsection {The resolutions of $I_{1\times k}(1,n)$}

Next we consider the case  $E=\BC$, $F=\BC^n$.
In this case $S=Sym(F)$ and the ideal $I_{1\times k}(1,n)$ is just the ideal generated by all monomials of degree $k$. 
The resolution of this ideal is well-known,
see,  e.g.,  \cite{MR3565359} or   \cite{eisenbud}.

The whole resolution is linear and $GL(F)$-equivariant, and its $k$-th  term is

\begin{equation}
\label{monres}
\mathbb {F}_j=  S_{( k,1^{j}) }F\otimes S(-k-j).
\end{equation}
 
\subsection{The resolution of  $\ci^{sqf;n,k}$}

We work over $S=\BC[x_1,\ldots ,x_n]=Sym(F)$.

Define a subcomplex $\mathbb {H}^{lin}$ of the complex $\mathbb{F}^{lin}$ given by   \eqref{monres}
by setting   $\underline{\mathbb{H}}_j^{lin}$ to be the subspace  of $\underline{\mathbb{F}}_j^{lin}$ spanned  by the basis elements of regular weight.
Note  that 
\[
\mathbb {H}^{lin}=\oplus_j \underline{\mathbb{H}}_j^{lin}\ot S(-k-j)
\]
is indeed a subcomplex of $\mathbb{F}^{lin}$.

\begin{theorem} \label{monothm}
The complex $\mathbb{H}^{lin}$ is the  linear part of the minimal free resolution of the ideal $\ci^{sqf;n,k}$.
We have the  $\FS_n$-module decomposition
\be\label{mkchow}
\underline{\mathbb{H}}_j^{lin}=Ind_{\tilde \FS_{\k+j}}^{\FS_n} \widetilde{[\k,1^{j }]}  , 
\ene
which has dimension $\binom{\k+j-2}{j-1}\binom{n}{\k+j}$.
\end{theorem}

\begin{proof}

We want to show that $\mathbb{H}^{lin}$ is the  linear strand of the resolution of  $\ci^{sqf;n,k}$. We proceed by induction on $j$. The case $j=0$ is clear because 
the generators of $\ci^{sqf;n,k}$ are precisely the generators of $I_{1\times k}(1,n)$ with regular weights.
Assume we proved the result for $j-1$ and consider  the $j$-th module in the resolution. As with the subpermanent case, it is enough to consider  the elements of a subregular weight as 
linear relations between elements of regular weight   are either of regular or subregular weight.

Consider  the syzygies in homological dimension $j$ and in homogeneous degree $k+j$ in term of cycles in the Koszul complex
$\BC(x_1,\ldots ,x_n; \ci^{sqf;n,k})$.
These will be cycles of the form
$$z=\sum_t a_t e_{a_{1,t}}\wedge\ldots \wedge e_{a_{j,t}}\otimes x_{u_{1,t}}x_{u_{2,t}}\ldots x_{u_{k,t}}$$
where the total weight is $(2,1^{k+j-2})$ and all monomials $x_{u_{1,t}}x_{u_{2,t}}\ldots x_{u_{k,t}}$ are of regular weights, and $a_t$ are scalars.
In each summand with $a_t\ne 0$ we are forced to have one $1$ among $a_{i,t}$ and one $1$ among $u_{i,t}$. 
So we can assume that in each summand with $a_t\ne 0$ we have $a_{1,t}=1$ and $u_{1,t}=1$.
We have $d(z)=0$. But, looking at the coefficient of $d(z)$ with respect to the basis vector $e_{a_{2,t}}\wedge\ldots \wedge e_{a_{j,t}}$ we see that its coefficient
is just $a_t x_1x_{u_{1,t}}x_{u_{2,t}}\ldots x_{u_{k,t}}$ which forces
$a_t$ to be zero. The dimension formula follows as in Theorem 3.1.
 \end{proof}

\begin{remark} The easiest way to see that the resolution of the ideal $\ci^{sqf;n,k}$ is linear and to see the ranks of the modules is
to observe that
for the $k\times n$ matrix
$$M(A, X)=\left(\begin{matrix}a_{1,1}x_1&\ldots&a_{1,n}x_n\\
\ldots&\ldots&\ldots\\
a_{k,1}x_1&\ldots&a_{k,n}x_n
\end{matrix}\right)$$
where 
$$A=\left(\begin{matrix}a_{1,1}&\ldots&a_{1,n}\\
\ldots&\ldots&\ldots\\
a_{k,1}&\ldots&a_{k,n}
\end{matrix}\right)$$
is a matrix of scalars with all maximal minors non-zero,
  the ideal of maximal minors of the matrix $M(A, X)$ is just
  $\ci^{sqf;n,k}$, so the resolution in question is   an
  Eagon-Northcott complex (\cite{eisenbud} or \cite[6.1.6]{weyman}).
After  this paper was submitted, a   characteristic free description of the resolution of $\ci^{sqf;n,k}$, with explicit differentials  appeared in     \cite{2016arXiv160906396G}.
\end{remark}
\pagebreak

\section{Additional results}
For $P\in S^d\BC^N$, let  $I^{P,k}\subset S^k\BC^N$ denote the ideal generated by the
partial derivatives  of $P$ of order $d-k$. 

\subsection{Size two subpermanents}
\begin{theorem}\label{tpermtwo} Let $I_t^{\tperm_n,2}$ denote the degree
$t$ component of the ideal generated by the size two sub-permanents of
an $n\times n$ matrix, so 
$\tdim I_2^{\tperm_n,2}={\binom n2}^2$. 
Then
\begin{align*}
\tdim \BC[x_{i,j}]/I^{perm_n,2}_t=&
\binom{n^2+t-1}t-
 [{\binom{n}{t}}^2+
n^2+(t-1)(\binom{n^2}2-{\binom n2}^2)+2\binom{t-1}2({\binom n2}^2+n\binom n3)
\\
&+2n\sum_{j=3}^{t-1} \binom{t-1}j\binom n{j+1}],
\end{align*}
Where recall that $\binom nt=0$ for $t>n$, in which case the
  formula is $\tdim S^t\BC^{n^2}$ minus the  value of the  Hilbert polynomial at $t$.
\end{theorem}

\def \ahal{{\mathrel{\smash-}}{\mathrel{\mkern-8mu}}
{\mathrel{\smash-}}{\mathrel{\mkern-8mu}} {\mathrel{\smash-}}{\mathrel{\mkern-8mu}}}

First, the Hilbert polynomial:
 
 \begin{theorem}
For the ideal $\ci^{\tperm_n,2}$ of $2\times 2$ permanents of
an $n \times n$ matrix, the Hilbert polynomial of $Sym(V)/\ci^{\tperm_n,2}$ is
\begin{equation}\label{HPSR}
\sum\limits_{i=0}^n f_i {t-1 \choose i},
\end{equation}
where $f_i$ is the $i^{th}$ entry in the vector
\[
\Big[ n^2, {n^2 \choose 2}-{n \choose 2}^2, 2{n \choose 2}^2+2n{n \choose 3},2n{n \choose 4} ,2n{n \choose 5}, \ldots, 2n{n \choose n} \Big].
\]
\end{theorem}
\begin{proof}
  \cite[Thm. 3.2]{MR1777172} gives a Gr\"obner basis for $ \sqrt{\ci^{\tperm_n,2}}$, the radical of $\ci^{\tperm_n,2}$, 
and by   \cite[Thm. 3.3]{MR1777172}, $\sqrt{\ci^{\tperm_n,2}}/\ci^{\tperm_n,2}$ has finite length, so vanishes
in high degree. The Hilbert polynomial only measures dimension asymptotically,
so
\[
\HP(Sym(V)/\sqrt{\ci^{\tperm_n,2}},t) = \HP(Sym(V)/\ci^{\tperm_n,2},t).
\]
By \cite{MR1777172}, for any diagonal term order,  the Gr\"obner basis
for $\sqrt{\ci^{\tperm_n,2}}$ is given by quadrics of the form
\[
x_{ij}x_{kl}+x_{kj}x_{il} \mbox{ with } i < k,\mbox{ }j< l,
\]
and five sets of cubic monomials
\be\label{cubics}
\begin{array}{ccc}
x_{i_1j_1}x_{i_1j_2}x_{i_2j_3} & i_1>i_2 & j_1<j_2<j_3 \\
x_{i_1j_1}x_{i_2j_2}x_{i_2j_3} & i_1>i_2 & j_1<j_2<j_3 \\
x_{i_1j_1}x_{i_2j_1}x_{i_3j_2} & i_1<i_2<i_3 & j_1>j_2 \\
x_{i_1j_1}x_{i_2j_2}x_{i_3j_2} & i_1<i_2<i_3 & j_1>j_2 \\
x_{i_1j_1}x_{i_2j_2}x_{i_3j_3} &  i_1<i_2<i_3 & j_1>j_2>j_3.
\end{array}
\ene
The key observation is that all the cubic monomials are square-free,
as  are the initial terms of the quadrics. Thus  the initial
ideal of $\sqrt{\ci^{\tperm_n,2}}$ is a square-free monomial ideal and  corresponds to the
Stanley-Reisner ideal of a simplicial complex $\Delta$. By  
\cite[Lemma 5.2.5]{MR2011360}, the Hilbert polynomial is as in Equation~\eqref{HPSR}, where $f_i$ is
the number of $i$-dimensional faces of $\Delta$. As the vertex set of
$\Delta$ corresponds to all lattice points $(i,j)$ with $1\le i,j \le n$,
it is immediate that $f_0=n^2$.

Since $x_{ij}x_{kl}$ is a non-face if $i<k$, $j<l$, no  edge
connects a southwest lattice point to a northeast lattice point.
Hence, the edges of $\Delta$ consist of all pairs $(i,j),(k,l)$ with
$i\geq k$ and $j \geq l$, of which there are ${n^2 \choose 2}-{n \choose 2}^2$.

Next, consider  the triangles of $\Delta$. Equation~\eqref{cubics} says there are
no triangles in $\Delta$ of the types in Figure~1. Also, there are 
no triangles which contain an edge connecting vertices at positions $(i,j)$ and $(k,l)$ 
with $i<k$,$j<l$. Thus, the only triangles in $\Delta$ are right triangles, but with
hypotenuse sloping from northwest to southeast. For a lattice point
$v$ at position $(d,e)$ there are exactly $(d-1)(e-1)$ right triangles
having $v$ as their unique north-most vertex. In the rightmost column $n$,
there are no such triangles, in the next to last column $n-1$ there are
$(n-1)+(n-2)+ \cdots = {n \choose 2}$ such triangles.
Continuing this way yields a total count of
\[
(n-1){n \choose 2} + (n-2){n \choose 2} + \cdots 2{n \choose 2}+{n \choose 2} = {n \choose 2}^2
\]
such right triangles, and taking into account the right triangles for which
$v$ is the unique south-most vertex doubles this number.
\begin{figure}\label{ntris}
\renewcommand{\thesubfigure}{\arabic{subfigure}}
\tikzstyle{dot}=[shape=circle,fill=black,draw=black,fill opacity=1,inner sep=.4 mm]

\begin{subfigure}[b]{.4\textwidth}
\centering
\begin{tikzpicture}[scale=1.2]
\draw[step=.8 cm] (-.2,-.2) grid (.8,1.6);
\draw (-.2,0) node[left]{$j_1$};
\draw (-.2,.8) node[left]{$j_2$};
\draw (-.2,1.6) node[left]{$j_3$};
\draw (0,-.2) node[below]{$i_3$};
\draw (.8,-.2) node[below]{$i_1=i_2$};
\filldraw[draw=black,fill=gray,fill opacity=.5] (.8,0) node[dot]{}--(0,.8) node[dot]{}--(0,1.6) node[dot]{}--cycle;
\end{tikzpicture}
\caption{}
\end{subfigure}
\begin{subfigure}[b]{.4\textwidth}
\centering
\begin{tikzpicture}[scale=1.2]
\draw[step=.8 cm] (-.2,-.2) grid (.8,1.6);
\draw (-.2,0) node[left]{$j_1$};
\draw (-.2,.8) node[left]{$j_2$};
\draw (-.2,1.6) node[left]{$j_3$};
\draw (0,-.2) node[below]{$i_2=i_3$};
\draw (.8,-.2) node[below]{$i_1$};
\filldraw[draw=black,fill=gray,fill opacity=.5] (.8,0) node[dot]{}--(.8,.8) node[dot]{}--(0,1.6) node[dot]{}--cycle;
\end{tikzpicture}
\caption{}
\end{subfigure}

\begin{subfigure}[b]{.4\textwidth}
\centering
\begin{tikzpicture}[scale=1.2]
\draw[step=.8 cm] (-.2,-.2) grid (1.6,.8);
\draw (-.2,0) node[left]{$j_3$};
\draw (-.2,.8) node[left]{$j_1=j_2$};
\draw (0,-.2) node[below]{$i_1$};
\draw (.8,-.2) node[below]{$i_2$};
\draw (1.6,-.2) node[below]{$i_3$};
\filldraw[draw=black,fill=gray,fill opacity=.5] (1.6,0) node[dot]{}--(.8,.8) node[dot]{}--(0,.8) node[dot]{}--cycle;
\end{tikzpicture}
\caption{}
\end{subfigure}
\begin{subfigure}[b]{.4\textwidth}
\centering
\begin{tikzpicture}[scale=1.2]
\draw[step=.8 cm] (-.2,-.2) grid (1.6,.8);
\draw (-.2,0) node[left]{$j_3$};
\draw (-.2,.8) node[left]{$j_1=j_2$};
\draw (0,-.2) node[below]{$i_1$};
\draw (.8,-.2) node[below]{$i_2$};
\draw (1.6,-.2) node[below]{$i_3$};
\filldraw[draw=black,fill=gray,fill opacity=.5] (1.6,0) node[dot]{}--(.8,0) node[dot]{}--(0,.8) node[dot]{}--cycle;
\end{tikzpicture}
\caption{}
\end{subfigure}

\begin{subfigure}[b]{\textwidth}
\centering
\begin{tikzpicture}[scale=1.2]
\draw[step=.8 cm] (-.2,-.2) grid (1.6,2.4);
\draw (-.2,0) node[left]{$j_3$};
\draw (-.2,.8) node[left]{$j_2$};
\draw (-.2,2.4) node[left]{$j_1$};
\draw (0,-.2) node[below]{$i_1$};
\draw (.8,-.2) node[below]{$i_2$};
\draw (1.6,-.2) node[below]{$i_3$};
\filldraw[draw=black,fill=gray,fill opacity=.5] (1.6,0) node[dot]{}--(.8,.8) node[dot]{}--(0,2.4) node[dot]{}--cycle;
\end{tikzpicture}\hspace{20 pt}
\begin{tikzpicture}[scale=1.2]
\draw[step=.8 cm] (-.2,-.2) grid (1.6,2.4);
\draw (-.2,0) node[left]{$j_3$};
\draw (-.2,1.6) node[left]{$j_2$};
\draw (-.2,2.4) node[left]{$j_1$};
\draw (0,-.2) node[below]{$i_1$};
\draw (.8,-.2) node[below]{$i_2$};
\draw (1.6,-.2) node[below]{$i_3$};
\filldraw[draw=black,fill=gray,fill opacity=.5] (1.6,0) node[dot]{}--(.8,1.6) node[dot]{}--(0,2.4) node[dot]{}--cycle;
\end{tikzpicture}
\caption{}
\end{subfigure}
\caption{Non-triangles of $\Delta$ from Equation~\eqref{cubics}}
\end{figure}
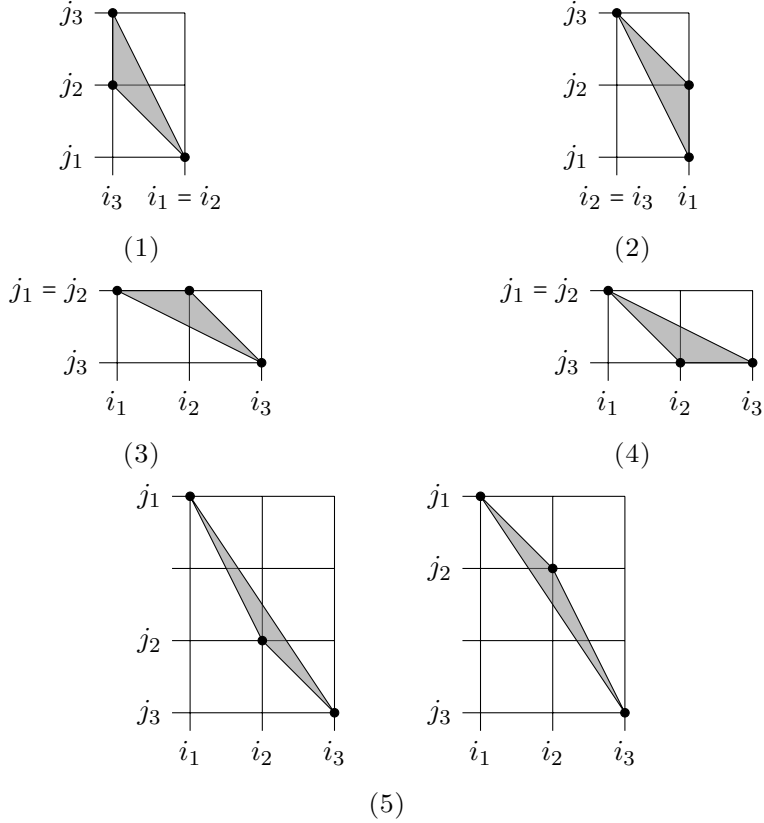

However, this count neglects thin triangles--those which have
all vertices in the same row or column. Since the number of thin triangles
is $2n{n \choose 3}$, the final count for the triangles of $\Delta$ is
\[
2{n \choose 2}^2+2n{n \choose 3}.
\]
For tetrahedra, the conditions of Equation~\eqref{cubics} imply that there
can only be thin tetrahedra, and an easy count gives $2n{n \choose 4}$ such.
The same holds for higher dimensional simplices, and concludes the proof.
\end{proof}
\begin{corollary}
For the ideal $\ci^{\tperm_n,2}$ of $2\times 2$ permanents of
an $n \times n$ matrix, the Hilbert function of $Sym(V)/\ci^{\tperm_n,2}$ is, when $t\geq 3$,
\begin{equation}\label{HFSR}
\HF(Sym(V)/\ci^{\tperm_n,2},t) = {n \choose t}^2 + \HP(Sym(V)/\ci^{\tperm_n,2},t),
\end{equation}
and it equals the Hilbert polynomial for $t>n$.
\end{corollary}
\begin{proof}
  The Hilbert function of
$\sqrt{\ci^{\tperm_n,2}}/\ci^{\tperm_n,2}$ in degree $t$ is ${n \choose t}^2$ by  \cite[Thm. 3.3]{MR1777172}. The result
follows by combining Equation~\ref{HPSR} with the short exact sequence
\[
0 \longrightarrow \sqrt{\ci^{\tperm_n,2}}/\ci^{\tperm_n,2} \longrightarrow Sym(V)/\ci^{\tperm_n,2} \longrightarrow Sym(V)/\sqrt{\ci^{\tperm_n,2}} \longrightarrow 0,
\]
and additivity of the Hilbert function.
\end{proof} 

For the purposes of comparing with other ideals, we rephrase this as: 

\begin{corollary}
$\tdim I_2^{\tperm_n,2}={\binom n2}^2$. 
For $t\geq 3$:
\begin{align*}
\tdim I^{perm_n,2}_t=&
\binom{n^2+t-1}t-
 [{\binom{n}{t}}^2+
n^2+(t-1)(\binom{n^2}2-{\binom n2}^2)+2\binom{t-1}2({\binom n2}^2+n\binom n3)
\\
&+2n\sum_{j=3}^{t-1} \binom{t-1}j\binom n{j+1}].
\end{align*}
\end{corollary}

\subsection{Hilbert functions for ideals of square-free monomials}
Although these can be deduced from our resolutions, we present the Hilbert functions and polynomials
for the ideals generated by square-free monomials.

\begin{proposition}\label{klimest} The Hilbert function of $\ci^{x_1\cdots x_n,\k}$  in degree $\k+t$ is
\be\label{chowhilb}
\tdim \ci^{(x_1\cdots x_n),\k}_{\k+t}
=\sum_{j=0}^{n-\k}
\binom n{\k-j}\binom {\k+t-1}{\k+j-1}
\ene
\end{proposition}

\begin{proof} The ideal in degree $d=t+\k$ has a basis of the distinct  monomials of degree $d$  containing at least $n-k$ distinct indices. When
we divide such a basis vector  by $x_1\cdots x_n$ the denominator will have degree at most $\k$. 
For each $i\leq \k$, the space of  possible numerators with a denominator of degree $i$
that is fixed, has dimension  $\tdim S^{d-n+i}\BC^{n-i}$, and there
are $\binom ni$ possible denominators. Summing over $i$ gives the result.
\end{proof}

For the Hilbert function  of
the coordinate ring, we have the following expression: 

\begin{proposition} The Hilbert function of $Sym(\BC^n)/\ci^{x_1\cdots x_n,\k}$ in degree $t$ is
\be\label{anotherexpr}
\tdim (Sym(\BC^n)/\ci^{x_1\cdots x_n,\k})_t
=\sum\limits_{j=0}^{n-\k-2}{n \choose j+1}{t-1 \choose j},
\ene
if $t \ge n-\k-1$, and ${n+t-1 \choose n-1}$ if $t < n-\k-1$.
\end{proposition}

\bibliographystyle{amsplain}
 
\bibliography{Lmatrix}

 \end{document}